\documentclass{amsart}

\usepackage{ot-reflector}
\usepackage[margin=1.4in]{geometry}
\usepackage{subfig}
\title{Intersection of paraboloids and application to Minkowski-type problems}

\author{Pedro Machado Manh\~{a}es de Castro}
\address{Centro de Inform\'{a}tica, Universidade Federal de Pernambuco, Brazil.}
\author{Quentin M\'erigot}
\address{Laboratoire Jean Kuntzmann, CNRS/Universit\'{e} de Grenoble, France}
\author{Boris Thibert}
\address{Laboratoire Jean Kuntzmann, Universit\'{e} de Grenoble, France}

\begin{document}

\begin{abstract}
In this article, we study  the
intersection (or union) of the convex hull of N confocal
paraboloids (or ellipsoids) of revolution. This study is motivated by
a Minkowski-type problem arising in geometric optics. We show that in
each of the four cases, the combinatorics is given by the intersection
of a power diagram with the unit sphere. We prove the complexity is
$O(N)$ for the intersection of paraboloids and $\Omega(N^2)$ for the
intersection and the union of ellipsoids. We provide an algorithm to
compute these intersections using the exact geometric computation
paradigm. This algorithm is optimal in the case of the intersection of
ellipsoids and is used to solve numerically the far-field reflector
problem.
\end{abstract}

\maketitle

\section*{Introduction}

The computation of intersection of half-spaces is a well-studied
problem in computational geometry, which by duality is equivalent to
the computation of a convex hull.  Similarly, the computation of
intersections or unions of spheres is also well studied and can be
done by using power diagrams \cite{aurenhammer:1987}.  In this
article, we study the computation and the complexity of the
intersection of the convex hull of confocal paraboloids of revolution,
showing that it is equivalent to intersecting a certain power diagram
with the unit sphere.  Union of convex hull of confocal paraboloids of
revolution, and intersection or union of convex hull of confocal
ellipsoids of revolution can be studied using the same tools.  These
studies are motivated by inverse problems similar to Minkowski problem
that arise in geometric optics. We show how the algorithm we developed
to compute the intersection of paraboloids are used to solve large
instances of one of these problems.

\subsection*{Minkowski-type problems.}
A theorem of Minkowski asserts that given a family of directions
$(y_i)_{1\leq i\leq N}$ and a family of non-negative numbers
$(\alpha_i)_{1\leq i\leq N}$, one can construct a convex polytope with
exactly $N$ facets, such that the $i$th facet has exterior normal
$y_i$ and area $\alpha_i$. Aurenhammer, Hoffman and Aronov
\cite{aurenhammer1998minkowski} studied a variant of this problem
involving power diagrams and showed its equivalence with the so-called
constrained least-square matching problem.  This article is motivated
by yet another problem of Minkowski-type that arises in geometric
optics, which is called the far-field reflector problem in the
literature \cite{caffarelli2008weak,caffarelli1999numerical}.  Recall
that a paraboloid of revolution is defined by three parameters: its
focal point, its focal distance $\lambda$ and its direction $y$. We
assume that all paraboloids are focused at the origin, and we denote
$P(y,\lambda)$ the convex hull of a paraboloid of revolution with
direction $y$ and focal distance $\lambda$. We will say in the
following that $P(y,\lambda)$ is a solid paraboloid.  Paraboloids of
revolution have the well-known optical property that any ray of light
emanating from the origin is reflected by the surface $\partial
P(y,\lambda)$ in the direction $y$. 
Assume first that one wants to send the light emited from the origin
in $N$ prescribed directions $y_1,\hdots,y_N$. From the property of a
paraboloid of revolution, this can be done by considering a surface
made of pieces of paraboloids of revolution whose directions are among
the $(y_i)$. In the far-field reflector problem, one would also like
to prescribe the amount of light $\alpha_i$ that is reflected in the
direction $y_i$. A theorem of Oliker-Caffarelli
\cite{caffarelli2008weak} ensures the existence of a solution to this
problem: there exist unique (up to a common multiplicative constant)
focal distances $\lambda_1,\hdots,\lambda_N$ such that the surface
$\partial ( \cap_{1\leq i \leq N} P(y_i,\lambda_i)) $ reflects exactly
the amount $\alpha_i$ in each direction $y_i$.  Other types of inverse
problems in geometric optics can be formulated as Minkowski-type
problems involving the union of confocal solid paraboloids, and the
union or intersection of confocal ellipsoids
\cite{oliker2003mathematical,kochengin1997determination}.

\subsection*{Contributions.} 
Motivated by these Minkowski-type problems, our goal is to compute the
union and intersection of solid confocal paraboloids and ellipsoids of
revolution. Using a radial parameterization, each of these computations
is equivalent to the computation of a decomposition of the unit sphere
into cells, that are not necessarily connected.  Our contributions are
the following:

\begin{itemize}
\item We show that each of the four types of cells can be computed by
  intersecting a certain power diagram with the unit sphere
  (Propositions~\ref{prop:pow}, \ref{prop:pow2} and \ref{prop:pow3}).
  The approach is similar to the
  computation of union and intersection of balls using power diagrams
  in \cite{aurenhammer:1987}, or to the computation of
  multiplicatively weighted power diagrams in $\Rsp^{d-1}$ using power
  diagrams in $\Rsp^d$ \cite{Boissonnat:2003:CCE:644108.644159}.
\item We show that the complexity bounds of these four diagram types
  are different. In the case of intersection of solid confocal
  paraboloids in $\Rsp^3$, the complexity of the intersection diagram
  is $O(N)$ (Theorem~\ref{thm-complexity}). This is in contrast with
  the $\Omega(N^2)$ complexity of the intersection of a power diagram
  with a paraboloid in $\Rsp^3$
  \cite{Boissonnat:2003:CCE:644108.644159}. In the case of the union
  and intersection of solid confocal ellipsoids, we recover this
  $\Omega(N^2)$ complexity (Theorem~\ref{thm:ellipsoids-n2}). Finally,
  the case of the union of paraboloids is very different from the case
  of the intersection of paraboloids. Indeed, in the latter case, the
  corresponding cells on the sphere are connected, while in the former
  case the number of connected component of a single cell can be
  $\Omega(N)$ (Proposition~\ref{prop:flower}). The complexity of the
  diagram in this case is unknown.
\item In Section~\ref{sec:computation}, we describe an algorithm for
  computing the intersection of a power diagram with the unit
  sphere. This algorithm uses the exact geometric computation paradigm
  and can be applied to the four types of unions and intersections. It
  is optimal for the union and intersection of ellipsoids, but its
  optimality for the case of intersection of paraboloids is open.
\item This algorithm is then used for the numerical resolution of the
  far-field reflector problem. Using a known optimal transport
  formulation \cite{wang2004design,glimm2003optical} and similar
  techniques to \cite{aurenhammer1998minkowski}, we cast this problem
  into a concave maximization problem in Theorem~\ref{th:OT}. This
  allows us to solve instances with up to 15k paraboloids, improving
  by several order of magnitudes upon existing numerical
  implementations \cite{caffarelli1999numerical}.
\end{itemize}


\section{Intersection of confocal paraboloids of revolution}
Because of their optical properties, finite intersections of solid paraboloids of revolutions with the same focal point play a crucial role in
an inverse problem called the far-field reflector problem. This
inverse problem is explained in more detail in Section~\ref{sec:application}. Here we study the computation and complexity of such an
intersection when the focal point lies at the origin. We call this type
of intersection a paraboloid intersection diagram.



\subsection{Paraboloid intersection diagram}
A paraboloid of revolution in $\Rsp^d$ with focal point at the origin
is uniquely defined by two parameters: its focal distance $\lambda$
and its direction, described by a unit vector $y$. We denote the
convex hull of such a paraboloid by $P(y,\lambda)$. The boundary
surface $\partial P(y, \lambda)$ can be parameterized in spherical
coordinates by the radial map $u\in \Sph^{d-1} \mapsto
\rho_{y,\lambda}(u)\ u$, where the function $\rho_{y,\lambda}$ is
defined by:
\begin{equation} 
\rho_{y,\lambda}: u \in \Sph^{d-1} \mapsto
\frac{\lambda}{1-\sca{y}{u}}.
\end{equation}
Given a family $Y=(y_i)_{1\leq i\leq N}$ of unit vectors and a family
$\lambda = (\lambda_i)_{1\leq i\leq N}$ of \emph{positive} focal distances, the
boundary of the intersection of the solid paraboloids
$(P(y_i,\lambda_i))_{1\leq i \leq N}$ is parameterized in spherical
coordinates by the function:
\begin{equation}
\rho_{Y,\lambda}(u) := \min_{1\leq i \leq N} \rho_{y_i,\lambda_i}(u) = \min_{y \in Y}  \frac{\lambda_i}{1-\sca{y_i}{u}}.
\end{equation} 

\begin{definition}
The \emph{paraboloid intersection diagram} associated to a family of
solid paraboloids $(P(y_i,\lambda_i))_{1\leq i\leq N}$ is a
decomposition of the unit sphere into $N$ cells defined by:
\[
\PI_{Y}^\lambda(y_i) := \{ u \in \Sph^{d-1};~ \forall j \in \{1,\hdots,N\}, \rho_{y_i,\lambda_i}(u) \leq \rho_{y_j,\lambda_j}(u)\}.
\]
The paraboloid intersection diagram corresponds to the decomposition
of the unit sphere given by the lower envelope of the functions
$(\rho_{y_i,\lambda_i})_{1\leq i \leq N}$.
\end{definition}

\subsection{Power diagram formulation}
We show in this section that each cell of the paraboloid intersection
diagram is the intersection of a cell of a certain power diagram with
the unit sphere. We first recall the definition of a \emph{power diagram}.
Let $P = (p_i)_{1\leq i\leq N}$ be a family of points in $\Rsp^d$ and
$(\omega_i)_{1\leq i\leq N}$ a family of weights.  The power diagram
is a decomposition of $\Rsp^d$ into $N$ convex cells, called power
cells, defined by 
\[
\Pow_{P}^{\omega}(p_i) :=\Big\{x\in \Rsp^d,\ \forall j\in \{1,\hdots,N\}\ \nr{x-p_i}^2 + \omega_i \leq \nr{x-p_j}^2+\omega_j\Big\}.
\]
\begin{proposition}
\label{prop:pow}
Let $(P(y_i,\lambda_i))_{1\leq i \leq N}$ be a family of confocal paraboloids. One has
\[
\forall i\in\{1,\hdots,N\}\quad \PI_{Y}^{\lambda}(y_i) = \Sph^{d-1} \cap \Pow_P^{\omega}(p_i),
\]
where $P=(p_i)_{1\leq i \leq N}$ and $(\omega_i)_{1\leq i\leq N}$ are defined by
$p_i = - (\lambda_i^{-1}/2) y_i$ and $\omega_i =  - \lambda_i^{-1} -
\lambda_i^{-2}/4$.
\end{proposition}

\begin{proof}
For any point $u\in \Sph^{d-1}$, we have the following equivalence :
\[
u \hbox{ belongs to } \PI_{Y}^\lambda(y_k)
\Longleftrightarrow 
k = \arg\min_{1\leq i\leq N} \frac{\lambda_i}{1-\sca{y_i}{u}}
\Longleftrightarrow 
k = \arg\max_{1 \leq i \leq N}
\lambda_i^{-1}-\sca{u}{\lambda_i^{-1}y_i}.
\]
An easy computation gives :
\[
\begin{array}{rl}
\max_{1 \leq i \leq N}
\lambda_i^{-1}-\sca{u}{\lambda_i^{-1}y_i} &= 
\max_{1 \leq i \leq N} \lambda_i^{-1} -
\nr{u + \frac{1}{2}\lambda_i^{-1} y_i}^2 +
\nr{u}^2 + \frac{1}{4} \nr{\lambda_i^{-1} y_i}^2 \\
&= \nr{u}^2 -
\min_{1\leq i \leq N}
\left(\nr{u + \frac{1}{2}\lambda_i^{-1} y_i}^2 - \lambda_i^{-1}
- \frac{1}{4}\lambda_i^{-2}
\right).
\end{array}
\]
This implies that a unit vector $u$ belongs to the paraboloid
intersection cell $\PI_{Y}^{\lambda}(y_i)$ if and only if it lies in
the power cell $\Pow_P^{\omega}(p_i)$.
\end{proof}

One can remark that the paraboloid intersection diagram does not
change if all the focal distances $\lambda_i$ are multiplied by the
same positive constant. This implies that the intersection \emph{with
  the sphere} of the power cells defined in the above proposition does
not change under a uniform scaling by $\lambda$ (even though the whole power
cells change).

\subsection{Complexity of the paraboloid intersection diagram in $\Rsp^3$}
In this section, we show that in dimension three,
the complexity of the paraboloids intersection
diagram is linear in the number of paraboloids.

\begin{theorem}\label{thm-complexity}
Let $(P(y_i,\lambda_i))_{1\leq i \leq N}$ be a family of solid
paraboloids of $\Rsp^3$.
Then the number of
edges, vertices and faces of its paraboloid intersection diagram is in
$O(N)$.
\end{theorem}
The proof of this theorem strongly relies on the following proposition,
which shows that each cell $\PI_{Y}^\lambda(y_i)$ can be transformed
into a finite intersection of discs, 
and is thus connected. Note that while it is stated only in dimension $3$,
this proposition holds in any ambient dimension.
\begin{proposition}\label{prop:lemma-2-1}
For any two solid paraboloids $P(y,\lambda)$ and $P(z,\mu)$, the
projection of the set $\mathcal{L} := (\partial P(y,\lambda)) \cap
P(z,\mu)$ onto the plane $\{y\}^\orth$ orthogonal to $y$ is a disc
with center and radius
\[ c[(y,\lambda),(z,\mu)] = 2\lambda \pi_y(z) / \| y - z\|^2 \qquad
r[(y,\lambda),(z,\mu)] = 2\sqrt{\lambda \mu}/\| y - z\|, \] 
where
$\pi_y$ denotes the orthogonal projection on $\{y\}^\orth$. Moreover,
given a solid paraboloid $P(y,\lambda)$, the following map is
one-to-one:
\[
\begin{array}{rl}
F_{(y,\lambda)}: (\Sph^{d-1}\setminus \{ y \}) \times \Rsp^+ &\to \{y\}^{\orth} \times \Rsp^+\\
(z,\mu) &\mapsto 
\left(c[(y,\lambda),(z,\mu)], r[(y,\lambda),(z,\mu)]\right).
\end{array}
\]
\end{proposition}
\begin{proof}
  The proof of the first half of this proposition can be found in
  \cite{caffarelli2008regularity}, but we include it here for the sake
  of completeness.  We first show that the orthogonal projection onto
  the plane $\{y\}^\orth$ of the intersection $\mathcal{L}'
  := \partial P(y,\lambda) \cap \partial P(z,\nu)$ is a
  circle. Without loss of generality, we assume that $y$ is the last
  basis vector $(0,\hdots,0,1)$.  Recall that a paraboloid of
  revolution $\partial P(y,\lambda)$ is defined implicitly by the
  relation $\nr{x} = \sca{x}{y} + \lambda$. Hence, any point $x$ in
  $\mathcal{L}'$ belongs to the hyperplane defined by $\sca{x}{z-y} =
  \lambda - \mu$. If we denote by $z'=\pi_y(z)$, $z_d=\sca{z}{y}$,
  $x'=\pi_y(x)$ and $x_d=\sca{x}{y}$, one has
\[
x_d = \frac{\sca{z'}{x'}}{1-z_d} + \frac{\mu - \lambda}{1-z_d}.
\] The surface $\partial P(y,\lambda)$ can be parameterized over the
plane $\{y\}^\orth$ by the equation $x_d = 
\nr{x'}^2 / 2 \lambda - \lambda / 2$.
Combining this with the relations $\nr{z'}^2 + z_d^2=1$ and $\nr{y-z}^2
= 2(1-z_d)$, we get
\[
\nr{ x' - \frac{2 \lambda}{\nr{y-z}^2} z'}^2 = \frac{4 \lambda \mu}{\nr{y-z}^2}.
\] 
We deduce that the projection of $\mathcal{L}'$ onto the plane
$\{y\}'$ is a circle of center $c=2\lambda z' / \| y - z\|^2$
and of radius $r=2\sqrt{\lambda \mu} / \| y - z\|$. Therefore,
the projection $\pi_y(\mathcal{L})$ is either the disc enclosed by
this circle or its complementary. In order to exclude the latter case,
we remark that the intersection $P(y,\lambda)\cap P(z,\mu)$ is a
compact set, because it is convex and does not include a ray (assuming
$y\neq z$). Hence, the projection $\pi_y(\mathcal{L}) \subseteq
\pi_y(P(y,\lambda)\cap P(z,\mu))$ is also compact, and therefore it is the disc of center $c$ and radius $r$.

Let us now show that the map $F_{(y,\lambda)}$ is one-to-one. For a
fixed positive $\mu$, let $c(z) = c([y,\lambda],[z,\mu])$. For every
point $z$ in $\Sph^{2} \setminus \{ \pm y\}$, denote $\pi_y^1(z) =
\pi_y(z) / \nr{\pi_y(z)}$. This point belongs to the unit circle in
$\{y\}^\orth$, which coincides with the equator $E_y$ of the sphere
$\Sph^{2}$ which is equidistant to the points $\{\pm y\}$.  Then,
given any constant-speed geodesic $z(.)$ such that $z(\pm 1) = \pm y$
and such that $z(0) = z_0 \in E_y$, i.e., $z(t) = \sin(t \pi/2) y +
\cos(t\pi/2) z_0$, the following formula holds
\[ c(z(t)) = 2\lambda\frac{\nr{\pi_y(z(t))}}{\nr{y-z(t)}^2} z_0 =
\lambda \frac{\cos(t\pi/2)}{1 - \sin(t\pi/2)} z_0.\] One easily checks
that the function $t \in[-1,1) \mapsto \frac{\cos(t\pi/2)}{1 -
  \sin(t\pi/2)}$ is increasing and maps $[-1,1)$ to $[0,+\infty)$.
The mapping $z \in \Sph^2\setminus\{y\} \mapsto c(z) \in \{y\}^\orth$
thus transforms bijectively every geodesic arc joining the points $-y$
and $y$ into a ray joining the origin to the infinity on the plane
$\{y\}^\orth$, and is therefore bijective. From the bijectivity of $c$
and the formula defining the radius, one deduces that the map
$F_{(y,\lambda)}$ is one-to-one.
\end{proof}


\begin{proof}[Proof of Theorem \ref{thm-complexity}]
Proposition \ref{prop:lemma-2-1} implies that the projection of the set
\[
{\mathcal{L}}_i = \{ \rho_{Y,\lambda}(u) u, u \in \PI_{Y}^\lambda(y_i)
\} 
= (\partial P(y_i, \lambda_i)) \cap \Big(\bigcap_{j\neq i} P(y_j,\lambda_j)\Big)
\] 
onto the plane orthogonal to $y_i$ is a
finite intersection of discs, and therefore convex. Since the surface
$\partial P(y_i,\lambda_i)$ is a graph over the plane $\{y_i\}^\orth$,
we deduce that $\mathcal{L}_i$ is connected. This implies that its
radial projection on the sphere, namely the paraboloid intersection
cell $\PI_{Y}^\lambda(y_i)$, is also connected.
We denote by $V$ (\textit{resp.} $E$, $F$) the number of vertices (\textit{resp.} edges,  faces). Since, each cell $\PI_{Y}^\lambda(y_i)$ is connected, the number of faces $F$ is bounded by $N$.
Moreover, since there are at least three incident edges for each vertex, we have that $3V \leq 2E$.
Then, by Euler's formula, we have that $V \leq 2F-4$ and $E \leq 3F-6$.
\end{proof}

Even though the complexity of the paraboloid intersection diagram is
$O(N)$, it can not be computed faster than $\Omega(N \log N)$, as
stated in the proposition below.  We first define the genericity
condition used in the statement of this proposition.
\begin{definition}
A family of solid paraboloids $(P(y_i, \lambda_i))_{1\leq i \leq N}$
in $\Rsp^3$ is called \emph{in generic position} if for any subset
$(i_k)_{1\leq k\leq 4}$ of $\{1,\hdots, N\}$, the intersection 
$\bigcap_{1\leq k\leq 4} \partial
P(y_{i_k},\lambda_{i_k})$ is empty.
\end{definition}
Remark that the intersection of four paraboloids $(\partial P(y_{i_1},
\lambda_{i_k}))_{1\leq k\leq 4}$ contains a point $x$ if and only if
the projection $u=x/\nr{x}$ of this point on the unit sphere satisfies
the equations
$\nr{u-p_{i_1}}^2+\omega_{i_1}=\hdots=\nr{u-p_{i_4}}^2+\omega_{i_4}$,
where the points $(p_i)$ and the weights $(\omega_i)$ are defined by
Proposition~\ref{prop:pow}. The genericity condition is then
equivalent to the condition that for any quadruple of weighted points
$(p_{i_k},\omega_{i_k})_{1\leq k \leq 4}$, the weighted circumcenter
does not lie on $\Sph^2$.

\begin{proposition}\label{prop:comp}
The complexity of the computation of the paraboloid intersection diagram is $\Omega(N\log(N))$ under the algebraic tree model, and even under an assumption of genericity.
\end{proposition}


\begin{proof}
We take a family of $N$ real numbers $(t_i)_{1\leq i \leq N}$. For every $i\in\{1,\hdots,N\}$, we put $\lambda_i=1$ and $y_i=\varphi(t_i)$ , where the map $\varphi:\Rsp\to \Rsp^3$ defined by $\varphi(t)=(\frac{t^2-1}{1+t^2},\frac{2t}{1+t^2},0)$ is a parameterization of the equator  $\Sph^2 \cap \{z=0\}$ from which we removed the point $(1,0,0)$. The family of paraboloids $(P(y_i,\lambda))_{1\leq i \leq N}$ is such that every cell $\PI_{Y}^\lambda(y_i)$ is delimited by two half great circles between the two poles, each of these half circles being shared by two cells.
We add the points $y_{N+1}=(1,0,0)$, $y_{N+2}=(0,1,0)$, $y_{N+3}=(-1,0,0)$, $y_{N+4}=(0,0-1,0)$,  $y_{N+5}=(0,0,1)$ and $y_{N+6}=(0,0,-1)$, so that $(P(y_i,\lambda))_{1\leq i \leq N+6}$ is in general position. More precisely, the four points  $y_{N+1},\hdots, y_{N+4}$  are added to ensure that every points of the equator is at a distance strictly less than $\sqrt{2}/2$ from $\{y_1,\hdots,y_{N+4}\}$. The cells of the two poles $y_{N+5}$ and $y_{N+6}$ then do not intersect the equator and we keep the property that there exists a cycle with the $N+4$ vertices of $\{y_1,\hdots,y_{N+4}\}$ in the dual of the paraboloid intersection diagram. Finding this cycle  then amounts to sorting the values $(t_i)_{1\leq i \leq N+4}$. The conclusion holds from the fact that a sorting algorithm has a complexity $\Omega(N\log(N))$ under the \emph{algebraic tree model}.
\end{proof}

\section{Other types of union and intersections} 
Other quadrics, such as the ellipsoid of revolution, or one sheet of a
two-sheeted hyperboloid of revolution can also be parametrized over a
unit sphere by the inverse of an affine map
\cite{oliker2010characterization}.  In this section, we study the
combinatorics of the intersection of solid 
ellipsoids, one of whose focal points lie at the origin. Note that
this intersection naturally appears in the near-field reflector
problem, where one wants to illuminate points in the space instead of directions 
(as in the far-field reflector problem) \cite{kochengin1997determination}.  Furthermore, for both the
ellipsoid and the paraboloid cases, we also study the union of the
convex hulls.  We show that in these three cases, the combinatorics is
still given by the intersection of a power diagram with the unit
sphere. However, the complexity might be higher. We
show that there exists configuration of $N$ ellipsoids whose
intersection and union
have complexity $\Omega(N^2)$. An algorithm that matches this lower bound is provided in Section~\ref{sec:computation}.



\subsection{Union  of confocal paraboloids of revolution}
The union of a family of solid paraboloids $(P(y_i,\lambda_i))_{1\leq
  i \leq N}$ is star-shaped with respect to the origin. Moreover, its
boundary can be parameterized by the radial function $u \in \Sph^{d-1}
\mapsto \left(\max_{1\leq i\leq N} \rho_{y_i,\lambda_i}(u)\right)
\cdot u.$ The \emph{paraboloid union diagram} is a decomposition of
the sphere into cells associated to the upper envelope of the
functions $(\rho_{y_i,\lambda_i})_{1\leq i \leq N}$:
\[
\PU_{Y}^\lambda(y_i) := \{ u \in \Sph^{d-1};~ \forall j \in \{1,\hdots,N\}, \rho_{y_i,\lambda_i}(u) \geq \rho_{y_j,\lambda_j}(u)\}.
\]
As before, these cells can be seen as the intersection of certain
power cells with the unit sphere.
\begin{proposition}\label{prop:pow2}
Given a family $(P(y_i,\lambda_i))_{1\leq i \leq N}$ of solid
paraboloids, one has for all $i$,
\[
\PU_{Y}^{\lambda}(y_i) = \Sph^{d-1} \cap \pow_P^{\omega}(p_i),
\]
where the points and weights are given by $p_i = \frac{1}{2} \lambda_i^{-1} y_i$ and $\omega_i =
\lambda_i^{-1} - \frac{1}{4}\lambda_i^{-2}.$
\end{proposition}
Proposition~\ref{prop:lemma-2-1} implies that for every $i$, the projection of
${\mathcal{L}}_i = (\partial P(y_i, \lambda_i)) \cap \partial
\left(\bigcup_{1\leq j \leq N} P(y_j,\lambda_j)\right)$ onto the plane
orthogonal to $y_i$ is a finite intersection of \emph{complements of
  discs}. In particular, this set does not need to be connected in
general, and neither does the corresponding 
cell. This situation can happen in practice (see Proposition \ref{prop:flower}). 
Consequently, one cannot use the connectedness argument as in
the proof of Theorem~\ref{thm-complexity} to show that the paraboloid
union diagram has complexity $O(N)$ in dimension three. Actually, the following
proposition shows that a unique cell may have $\Omega(N)$ distinct connected component.

\begin{proposition}\label{prop:flower} One can construct a family of paraboloids
  $(P(y_i,\lambda_i))_{0\leq i\leq N}$ such that the paraboloid union
  cell $\PU_Y^\lambda(y_0)$ has $\Omega(N)$ connected components.
\end{proposition}

\begin{proof} Let $y_0$ be an arbitrary point on the sphere, and let
  $\lambda_0=1$. Now, consider a family of disks $D_i$ in the plane
  $H = \{y_0\}^\orth$ with centers and radii $(c_i,r_i)_{1\leq i \leq N}$,
  and such that the set
\[ U= \bigcup_{i=1}^N (H\setminus D_i) = H  \setminus \bigcup_{i=1}^N D_i\]
has $\Omega(N)$ connected components. This is possible by setting up a
flower shape (see Figure \ref{fig:flower}), i.e., $D_1$ is the unit ball and $D_2,\hdots,D_N$ are set
up in a flower shape around $D_1$. By the second part of
Proposition~\ref{prop:lemma-2-1}, one can construct paraboloids
$(P(y_i,\lambda_i))_{1\leq i\leq N}$ such that
$F_{(y_0,\lambda_0)}(y_i,\lambda_i) = (c_i,r_i)$. Then, the first part
of Proposition~\ref{prop:lemma-2-1} shows that the paraboloid union
cell $\PU_{Y}^\lambda(y_0)$ is homeomorphic to $U$, and has therefore
$\Omega(N)$ connected components.
\end{proof}
\begin{figure}[t]
  \centering          
 \includegraphics[width=0.2\textwidth]{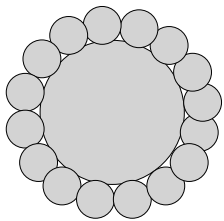}
  \caption{\textbf{A flower} }  \label{fig:flower}
\end{figure}

One can also underline that the complexity of each cell is $O(N)$. This is a direct consequence of 
Proposition~\ref{prop:lemma-2-1} and the fact that the complexity of the union of $N$ planar 
discs is $O(N)$ \cite[Lemma 1]{aurenhammer:1987}.

\subsection{Intersection and union of confocal ellipsoids of revolution}
An ellipsoid of revolution whose one focal point lies at the origin is
characterized by two other parameters: its second focal point $y$ and
its eccentricity $e$ in $(0,1)$. We denote the convex hull of such an
ellipsoid of revolution $E(y,e)$. The surface $\partial E(y,e)$ of
this set is parameterized in spherical coordinates by the function
\[
\sigma_{y,e}(m) := \frac{d}{1-e \sca{m}{\frac{y}{\|y\|}}} \quad \mbox{where } d=\frac{\|y\| (1-e^2)}{2 e}.
\] Note that the value $d$ is fully determined by $e$ and $y$ and is
introduced only to simplify the computations.

Let $Y=(y_i)_{1\leq i \leq N}$ be a family of distinct points in
$\Rsp^d$ and $e=(e_i)_{1\leq i \leq N}$ be a family of real numbers in
the interval $(0,1)$. The boundary of the intersection of solid ellipsoids
$\bigcap_{1\leq i\leq N} E(y_i, e_i)$ is parameterized in spherical
coordinates by the lower envelope of the functions
$(\sigma_{y_i,e_i})_{1\leq i \leq N}$. The \emph{ellipsoid
  intersection diagram} of this family of ellipsoids is the
decomposition of the unit sphere into cells associated to the lower
envelope of the functions $(\sigma_{y_i,e_i})_{1\leq i \leq N}$:
\begin{equation}
\EI_{Y}^e(y_i) := \{ u \in \Sph^{d-1};~ \forall j \in \{1,\hdots,N\}, \sigma_{y_i,e_i}(u) \leq \sigma_{y_j,e_j}(u)\}.
\label{eq:pow3}
\end{equation}
Similarly, the cells of the ellipsoid union diagram are associated to
the upper envelope of the functions $(\sigma_{y_i,e_i})_{1\leq i \leq
  N}$ as follows:
\begin{equation}
\EU_{Y}^e(y_i) := \{ u \in \Sph^{d-1};~ \forall j \in \{1,\hdots,N\}, \sigma_{y_i,e_i}(u) \geq \sigma_{y_j,e_j}(u)\}.
\label{eq:pow4}
\end{equation}
As in the case of paraboloids, the computation of each diagram amounts
to compute the intersection of a power diagram with the unit sphere.
\begin{proposition}\label{prop:pow3}Let $(E(y_i,e_i))_{1\leq i \leq N}$ be a family of solid confocal ellipsoids. Then,
\begin{itemize}
\item[(i)] The cells of the ellipsoid intersection diagram are given by
$\EI_{Y}^{e}(y_i) = \Sph^{d-1} \cap \pow_P^{\omega}(p_i),$
where 
$p_i = -\frac{e_i}{2d_i} \frac{y_i}{\|y_i\|}$
and $\omega_i = -\frac{1}{d_i} -
\frac{e_i^2}{4d_i^2}.$
\item[(ii)]
The cells of the ellipsoid union diagram are given by :
$\EU_{Y}^{e}(y_i) = \Sph^{d-1} \cap \pow_P^{\omega}(p_i),$
where 
$p_i = \frac{e_i}{2d_i} \frac{y_i}{\|y_i\|}$
and
$\omega_i =\frac{1}{d_i} -
\frac{e_i^2}{4d_i^2}.
$
\end{itemize}
\end{proposition}
The theorem below shows that in dimension three the complexity of these diagrams
can be quadratic in the number of ellipsoids. This is in sharp contrast with
the case of the paraboloid intersection diagram, where the complexity
is linear in the number of paraboloids. 
\begin{theorem}\label{thm:ellipsoids-n2}
  In $\Rsp^3$, there exists a configuration of confocal ellipsoids of
  revolution such that the number of vertices and edges in the
  ellipsoid intersection diagram (resp. the ellipsoids union diagram)
  is $\Omega(N^2)$.
\end{theorem}

The proof of this theorem strongly relies on the following lemma, that
transfers the problem to a complexity problem of power diagrams
intersected by the unit sphere.

\begin{lemma}\label{lemma:ellipsoids-powerdiagrams}
Let $(p_i,\omega_i)_{1\leq i \leq N}$ be a family of weighted points
such that no point $p_i$ lie at the origin. Then, there exists:
\begin{itemize}
\item[(i)] a family of ellipsoids $(E(y_i,e_i))_{1\leq i \leq N}$, such that for all $i$,
$\EI_{Y}^{e}(y_i) = \Sph^{d-1} \cap \pow_P^{\omega}(p_i).$ 
\item[(ii)] a family of  ellipsoids $(E(y_i,e_i))_{1\leq i \leq N}$, such that for all $i$,
$\EU_{Y}^{e}(y_i) = \Sph^{d-1} \cap \pow_P^{\omega}(p_i).$
\end{itemize}
\end{lemma}
\begin{proof}
We prove only the first assertion, the second one being similar. Let $i\in\{1,\hdots, N\}$. By inverting the equations
$p_i = -\frac{e_i}{2d_i} \frac{y_i}{\|y_i\|}$ and
$\omega_i =\frac{1}{d_i} +
\frac{e_i^2}{4d_i^2}
$,
we get
\[
y_i=\frac{-4}{(\omega_i +\nr{p_i}^2)^2 - 4 \nr{p_i}^2}\ p_i
\quad  \mbox{and} \quad
e_i=\frac{-2\nr{p_i}}{\omega_i + \nr{p_i}^2}.
\] The condition $e_i \in (0,1)$ is equivalent to $\omega_i < -
\nr{p_i}^2 \mbox{ and } \omega_i < -
\nr{p_i}^2 - \nr{p_i} $. These
inequalities can always be satisfied by substracting a large constant to all
the weights $\omega_i$, an operation that does not change the power
cells.
\end{proof}

We recall that the Voronoi diagram of a point cloud $(p_i)_{1\leq i \leq N}$ of $\Rsp^3$ is the decomposition of the space in $N$ convex cells defined by 
\[
\mathrm{Vor}_P(p_i):=\{x \in \Rsp^3,\ \nr{x-p_i} \leq \nr{x-p_j} \}.
\] 
The next proof is illustrated by Figure~\ref{fig:omega_n2}.

\begin{proof}[Proof of Theorem \ref{thm:ellipsoids-n2}]
  Thanks to Lemma \ref{lemma:ellipsoids-powerdiagrams}, it is
  sufficient to build an example of Voronoi diagram whose intersection
  with $\Sph^2$ has a quadratic number of edges and vertices. We can
  consider without restriction that $N=2k$ is even. Let
  $\varepsilon>0$ be a small number. We let $p_1,\hdots,p_k$ be $k$
  points uniformly distributed on the circle centered at the origin in
  the plane $\{z=0\}$, and with radius $2-\varepsilon$. We also
  consider $k$ evenly distributed points $q_1,\hdots,q_k$ on
  $(A,B)\setminus \{O\}$, where $A=(0,0,-\varepsilon/4)$,
  $B=(0,0,\varepsilon/4)$ and $O=(0,0,0)$.  Our goal is now to show
  that for any $i$ and $j$ in between $1$ and $k$, the intersection
  $\Vor_P(p_i) \cap \Vor_P(q_j) \cap \Sph^2$ is non-empty, where $P =
  \{p_i\}\cup \{q_j\}$. We denote by $m_{i,j}$ the unique point which
  is equidistant from $p_i$ and $q_j$, and which lies in the
  horizontal plane containing $q_j$ and in the (vertical) plane
  passing through $p_i$, $q_j$ and the origin. A simple computation
  shows that the distance $\delta_{i,j}$ between $m_{i,j}$ and $q_{j}$
  satisfies $\delta_{i,j} \leq 1-\varepsilon / 2 +
  \varepsilon^2/(64-32\varepsilon)$. Taking $\varepsilon$ small
  enough, this implies that the point $m_{i,j}$ belongs to the unit
  ball. The line passing through $m_{i,j}$ and orthogonal to the plane
  passing through $p_i$, $q_j$ and the origin cuts the unit sphere at
  two points $r^\pm_{i,j}$. These two points belong to the same
  horizontal plane as $q_j$, and they both belong to
  $\Vor_P(q_j)$. Moreover, the distance between $m_{i,j}$ and each of
  these two points is less than $(1-\delta_{i,j}^2)^{1/2}$.  We can
  take $\varepsilon$ small enough so that this distance is less than
  $\sin(\pi/k)$.  This implies that the points $r_{i,j}^\pm$ also
  belong to the Voronoi cell $\mathrm{Vor}(p_i)$ and therefore to the
  intersection $\Vor_P(p_i) \cap \Vor_P(q_j) \cap \Sph^2$. Thus, the
  intersection between the Voronoi diagram and the sphere has at least
  $k^2=\frac{1}{4}N^2$ spherical edges.
\end{proof}

\begin{figure}[t]
  \centering          
  \includegraphics[width=0.9\textwidth]{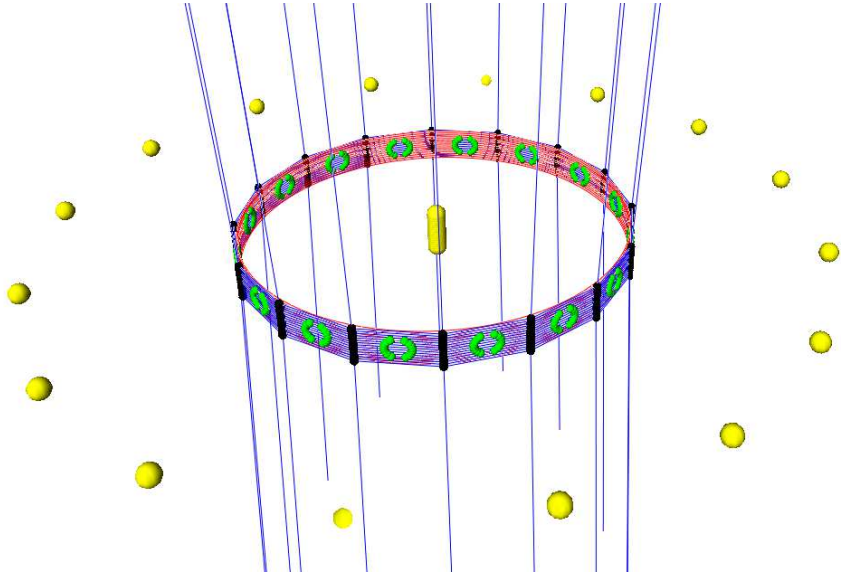} 
  \caption{\textbf{An $\Omega(N^2)$ construction.} {\small  \emph{From a total of $N$ sites, half of them in a segment inside the sphere, and half of them on a circle outside the sphere, we obtain a total of $\Omega(N^2)$ edges embedded on the sphere. } }  \label{fig:omega_n2}}
\end{figure}

\section{Computing the intersection of a power diagram with a sphere}
\label{sec:computation}

In this section, we describe a robust and efficient algorithm to
compute the intersection between the power diagram of weighted points
$(P,\omega)$ and the unit sphere. We call this 
\emph{intersection diagram}. Given the set of weighted points
$(p_i,\omega_i)$, we define the $i$th intersection cell as
$\ID_P^\omega(p_i) = \Sph^{2} \cap \Pow_P^\omega(p_i)$.  This
algorithm is implemented using the Computational Geometry Algorithms
Library (\textsc{CGAL})~\cite{cgal}, and bears some similarity to
\cite{Boissonnat:2003:CCE:644108.644159,Delage_cgal-basedfirst}.  We
concentrate our efforts in obtaining the exact combinatorial structure
of the diagram on the sphere for weighted points whose coordinates and
weight are rational. Exact constructions can be obtained, but are not
our focus, since in our applications, the results of this section are
used as input of a numerical optimization procedure; see
Section~\ref{sec:application}.

Note that the intersection diagram has 
complex combinatorics. If one assumes that the weighted points are
given by a paraboloid intersection diagram,
Proposition~\ref{prop:lemma-2-1} shows that the intersection cells are
connected. But even in this case, the intersection between two
intersection cells $\ID_P^\omega(p_i)$ and $\ID_P^\omega(p_j)$ can
have multiple connected components. Consequently, one cannot hope to
be able to reconstruct these cells from the adjacency graph (i.e., the
graph that contains the points in $P$ as vertices, and where two
vertices are connected by an arc if the two cells intersect in a
non-trivial circular arc), even in this simple case. In general, the
cells of the intersection diagram can be disconnected and they can
have holes, as shown in
Figure~\ref{fig:example_of_diagrams_on_sphere}. In the next paragraph,
we propose an algorithm to compute a boundary representation of these
cells. Recall the following definitions concerning power diagrams:

\begin{definition}
When two power cells $\Pow_P^\omega(p_{i})\cap \Pow_P^\omega(p_{j})$
intersect in a $2$-dimensional set, that is a (possibly unbounded)
polygon in $\Rsp^3$, they determine a \emph{facet} of the power
diagram. Similarly, a \emph{ridge} is a $1$-dimensional intersection
$\Pow_P^\omega(p_{i})\cap \Pow_P^\omega(p_{j}) \cap
\Pow_P^\omega(p_{k})$ between three power cells. A ridge can be either
a ray or a segment. The boundary of each facet can be written as a
finite union of ridges.
\end{definition}

\begin{figure}[t]
  \centering          
  \subfloat[general]{\label{fig:example_of_diagrams_on_sphere:1}\includegraphics[width=0.35\textwidth]{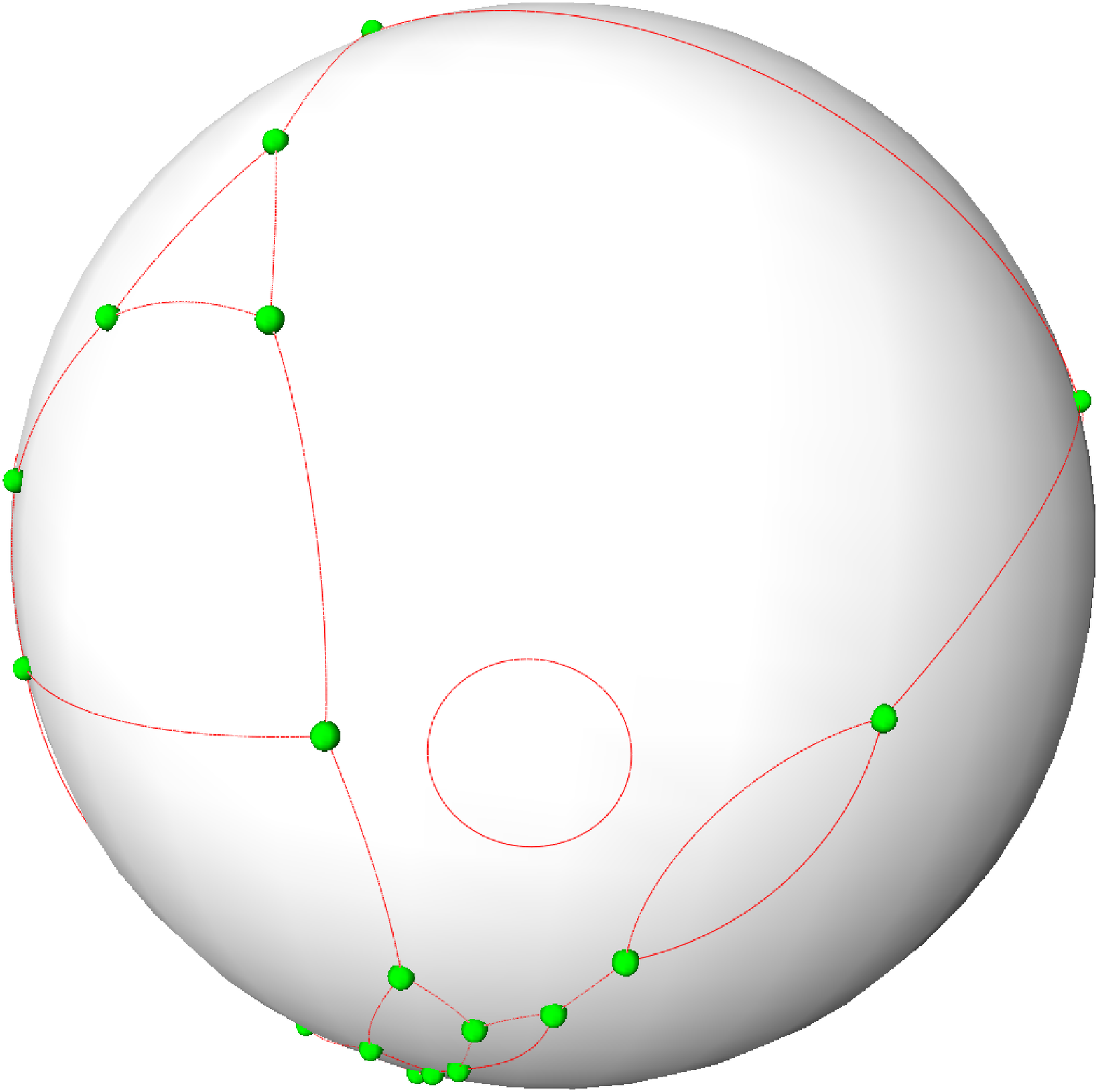}}
  \subfloat[cube]{\label{fig:example_of_diagrams_on_sphere:2}\includegraphics[width=0.4156\textwidth]{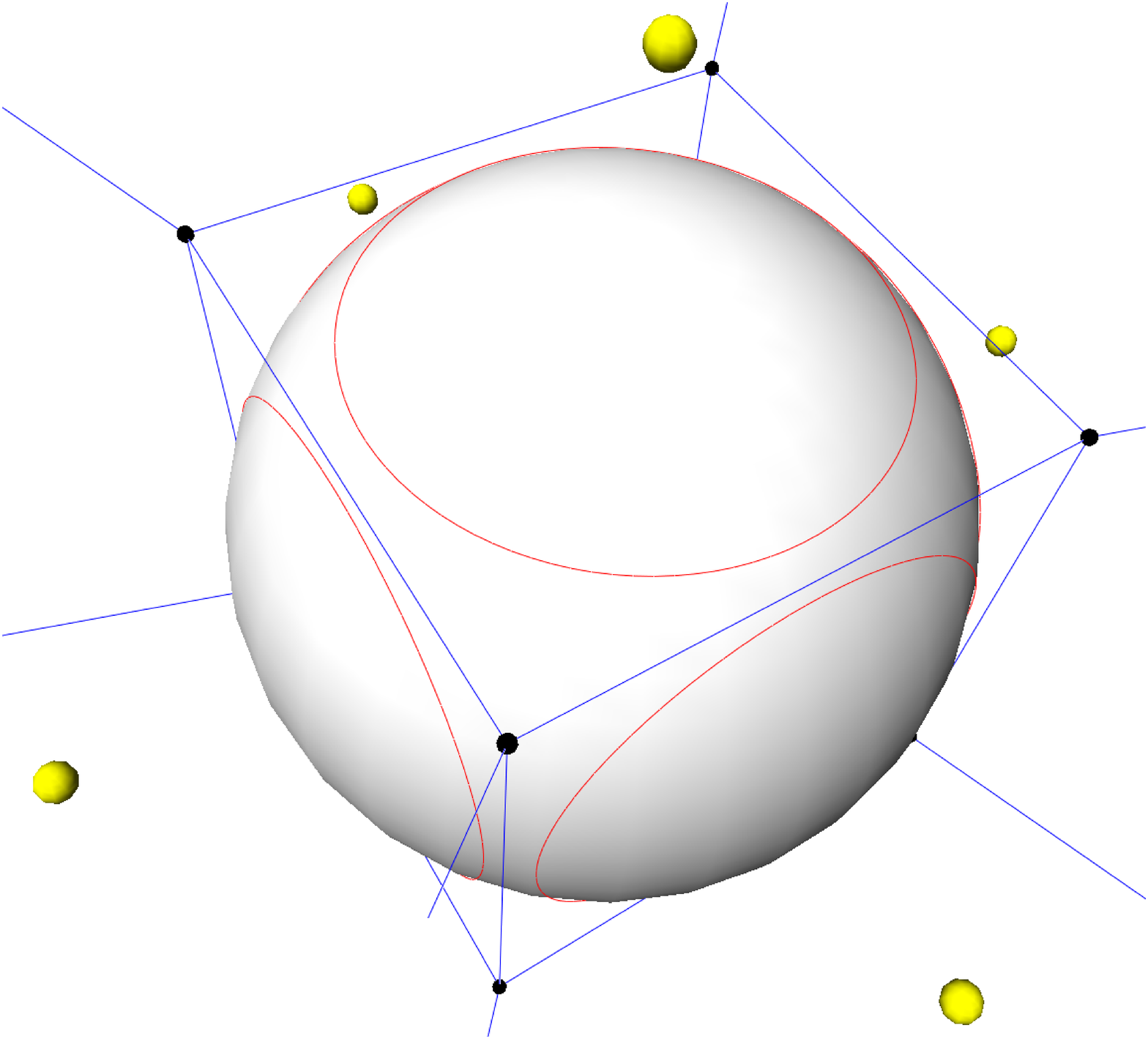}} 
  \caption{\textbf{Examples of diagram on the sphere.} {\small  \emph{Left: an intersection diagram containing faces with no vertices, faces with holes and faces with two vertices or more. Right: a diagram corresponding to the intersection of a cube and a sphere; there are seven faces, six faces with no vertices, and a face with no vertices and six holes.} }  \label{fig:example_of_diagrams_on_sphere}}
\end{figure}

In the 
remainder of this section, we explain how to compute a single
intersection cell $\ID_P^\omega(p_i)$, for a given $p_i$, without any
assumption on the points and weights.
This intersection cell can be quite intricate (multiple connected
components, holes), and will therefore be represented by its
\emph{oriented boundary}. More precisely, the boundary
$\ID_P^\omega(p_i)$ is a finite union of closed curves called
\emph{cycles}. These cycles are oriented so that for someone walking
\emph{on} the sphere following a cycle, the intersection cell
$\ID_P^\omega(p_i)$ lies to the right. The cell is uniquely determined
by its oriented boundary. The following theorem is the main result of this section.

\begin{theorem}\label{thr:alg-complexity}
There is an $O(N \log N + C)$ algorithm for obtaining the intersection
diagram of a set $P$ of $N$ weighted points in $\mathbb{R}^3$, where
$C$ is the complexity of the power diagram of $P$.
\end{theorem}

From Theorem~\ref{thm:ellipsoids-n2}, the size of the output in the
worst case for the intersection (or union) of confocal ellipsoids is
$\Omega(N^2)$. This implies that the algorithm described below is
optimal for computing an ellipsoid intersection (or union)
diagram. The optimality of the algorithm for the paraboloid
intersection diagram is an open problem, which can be phrased as
follows: 

\begin{openquestion}
Is the complexity of the power diagram given by Proposition \ref{prop:pow} bounded by $O(N)$ ?
\end{openquestion}

\paragraph{Remark.} Even if the question above has a negative answer, 
at least for the particular case of the intersection diagram of confocal paraboloids, there exists an
easy randomized incremental algorithm attaining the lower bound in Proposition~\ref{prop:comp}.
If the input is $N$ confocal paraboloids, from Theorem~\ref{thm-complexity}, we know that
the complexity of the paraboloid intersection diagram is $O(N)$. Moreover the resulting cells are connected.
Adopting the \emph{randomized incremental construction} paradigm (see \cite{Sack:2000:HCG:337150,Mulmuley} 
for a comprehensive study), this means that the expected complexity of the cell of the $i$th inserted paraboloid 
is $O(N)/i$. Besides, thanks to the connectedness of the cells, we can detect the conflicts with 
a breadth-first search in linear time as well,  without any additional structure. 
Therefore, the total expected cost is $\sum_{i=1}^n O(N)/i$, which is in $O(N \log{N})$. 
Notice that this analysis relies on the fact that the cells are connected, which is not necessarily true in the case of union of paraboloids 
and in the ellipsoids cases.
In our application, we always use the generic and robust algorithm based on the intersections of power cells, using \textsc{CGAL}.



\subsection{Predicates}
Geometric algorithms often rely on predicates, i.e., estimation of
finite-valued geometric quantities such as the orientation of a
quadruple of points in $3$D, or the number of intersections between a
straight line and a surface.  These predicates need to be evaluated
exactly: if this is not the case, some geometric algorithm 
may even not terminate~\cite{kettneratal:2008}.  We adopt the
\emph{dynamic filtering} technique for the computation of our
predicates in \textsc{CGAL}. This means that for every arithmetic operation, we
also compute an error bound for the result. This can be automatized
using interval arithmetic~\cite{bronnimannetal:2001}. If the error
bound is too large to determine the result of the predicate, the
arithmetic operation is performed again with exact rational
arithmetic, allowing us to return an exact answer for the predicate.
Our algorithm requires the following three predicates:
\begin{itemize}
	\item[(i)] \texttt{has\_on(Point p)} returns $+1$, $0$, or
          $-1$, if the point is outside, on, or inside the unit sphere
          respectively. 
	\item[(ii)] \texttt{number\_of\_intersections(Ridge r)}
          returns the number of intersections between $r$ and the unit
          sphere ($0,1$ or $2$), where $r$ can be a segment or a
          ray. By convention, segment or ray touching the sphere
          tangentially has two intersection points, with same
          coordinates. Remark that the result of the first predicate
          can be obtained directly from the intermediate computations
          of this predicate. As we use them intertwined, this shortcut
          is adopted in our implementation.
	\item[(iii)] \texttt{plane\_crossing\_sphere(Facet f)} returns the number of intersections between the supporting plane of $f$ and the sphere ($0,1$ or $+\infty$). One may notice that this predicate is just the dual of the first one.
\end{itemize}

Even though the algorithm below is presented as working directly on
the cells of the power diagram, our actual implementation uses \textsc{CGAL},
and we work with the regular triangulation of the set of points
(that is, the dual of the power diagram). This triangulation is a simplicial complex with
$0$-, $1$-, $2$- and $3$-simplices. In \textsc{CGAL}, only the $0$- and
$3$-simplices are actually stored in memory, but it remains
possible to traverse the whole structure in linear time in the number
of simplices. In order to handle the boundary simplices, \textsc{CGAL} takes
the usual approach of adding a \emph{point at infinity} to the initial
set of points. In practice, this means the following. A ridge of the
power digram can be either a segment or a ray, and is dual of a
$2$-simplex in the triangulation. In \textsc{CGAL}, such a ridge corresponds to
two adjacent tetrahedra, one of whom may contain the point at
infinity. The predicates above need to be adapted in order to handle
$k$-simplices that are incident to the point at infinity. The general
predicates have been implemented, but the details are omitted here.




\subsection{Computation of the oriented boundary}
The output of our algorithm is the oriented boundary of the
intersection cell $\ID_P^\omega(p_i)$. Note that this is a purely
combinatorial object, and no geometric construction are performed
during its computation. It is described as a sequence of
\emph{vertices on the sphere}, \emph{oriented arcs}, and
\emph{oriented cycles}.

\begin{figure}[t]
  \label{fig:obtaining_arcs}
  \centering          
  \subfloat[inside]{\label{fig:obtaining_arcs:inside}\includegraphics[width=0.25\textwidth]{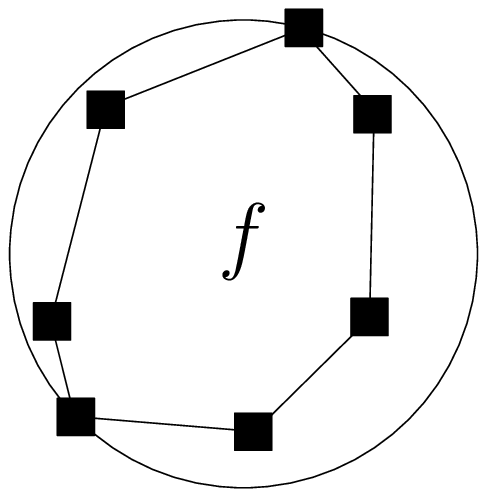}}
  \hspace{0.5cm}
  \subfloat[intersecting]{\label{fig:obtaining_arcs:intersecting}\includegraphics[width=0.25\textwidth]{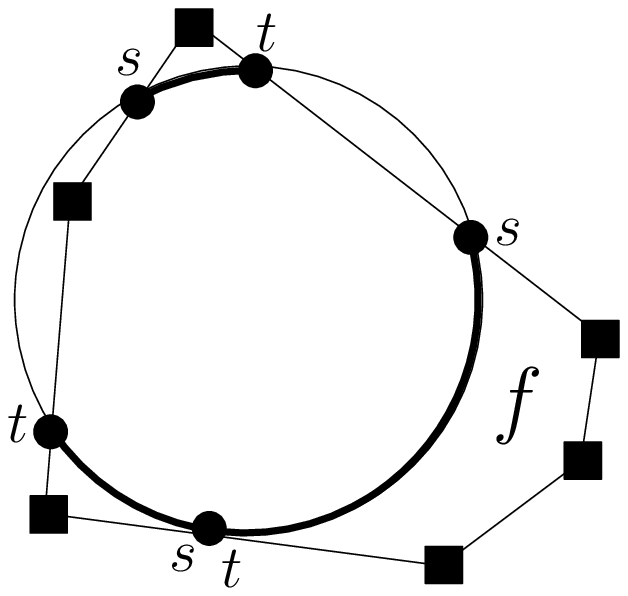}} 
  \hspace{0.5cm}
  \subfloat[enclosing]{\label{fig:obtaining_arcs:enclosing}\includegraphics[width=0.25\textwidth]{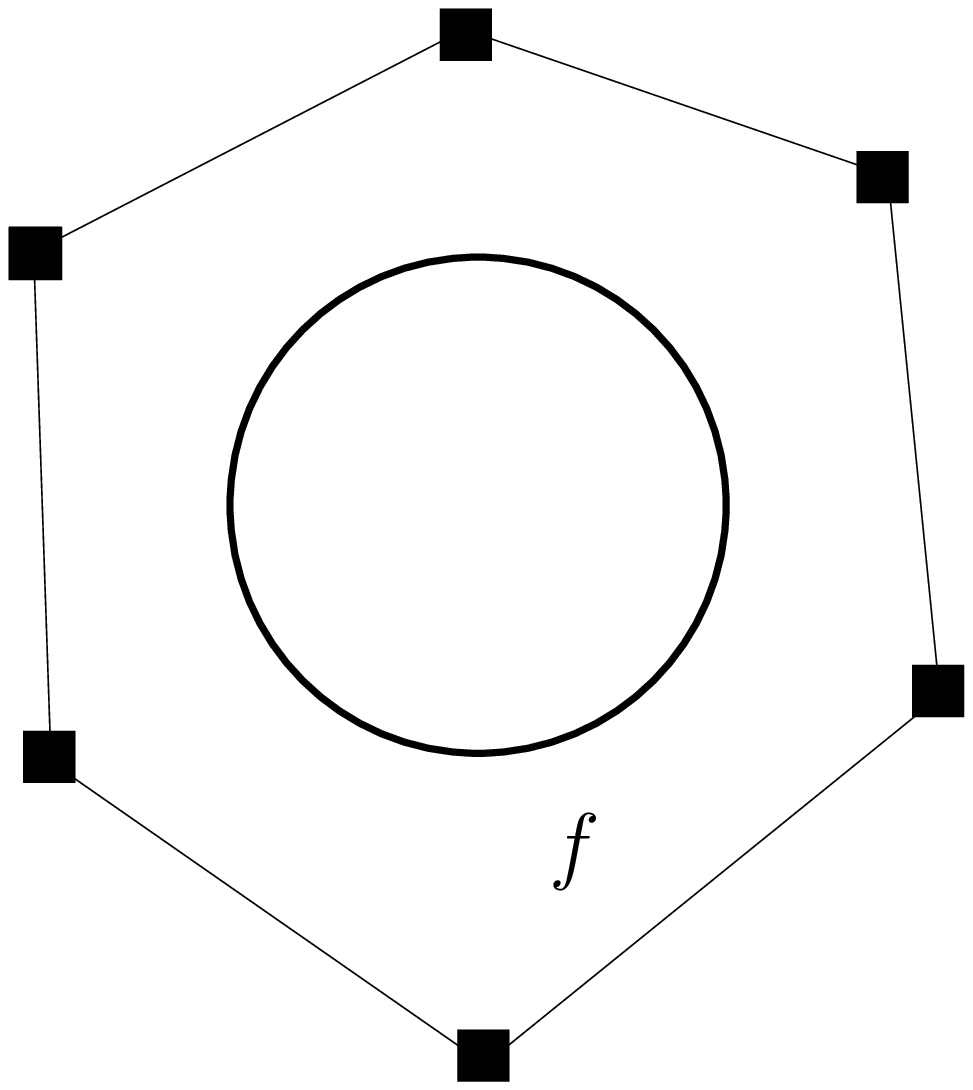}}
  \caption{\textbf{Obtaining the arcs of $\ID_P^\omega(p_i)$.} {\small  \emph{Three distinct cases: (a) inside the circle, (b) intersecting the circle, and (c) enclosing the circle. There is still a case not depict above, $f$ can be outside the circle as well. The arcs with zero length, such as the one in (b) can be removed from the diagram depending on its use.} }}
\end{figure}

$\bullet$ A \textbf{vertex on the sphere} is the result of an
intersection between a ridge of the power diagram and the unit sphere.
A vertex on the sphere is not constructed explicitly, but stored as an
ordered pair of extremities of the corresponding ridge. Since a ridge
$[x_1,x_2]$ can intersect the sphere twice, the order determines which
intersection point to consider, i.e., the pair $(x_1,x_2)$ describes
the intersection point that is \emph{closest} to $x_1$. This
convention allows us to handle infinite ridges, and to uniquely
identify vertices on the sphere. In degenerate cases, vertices may
correspond to different ordered pairs, but having the same
coordinates.

$\bullet$ An \textbf{oriented arc} corresponds to the (oriented)
intersection between $\Pow_P^\omega(p_i)$, another cell
$\Pow_P^\omega(p_j)$ and the unit sphere. Two situations
arise. Either the arc is a complete circle, in which case it is also a
cycle. Or the two extreme points are vertices as defined above,
that is, they are the intersection between ridges and the unit sphere. An arc is
oriented so that, when walking on the sphere from the first point to
the second, the cell lies on the right.



$\bullet$ An \textbf{oriented cycle} is a connected component of the
oriented boundary of $\ID_P^\omega(p_i)$. A cycle can have two or more
vertices on the sphere, in this case, they are represented by a cyclic
sequence of arcs. However, a cycle can also have no vertices, and in
such a case, it is represented by a single full arc. The orientation
of a cycle is given by the orientation of its arcs.

Notice that the boundary of a power cell is composed of several convex
facets, that can possibly be unbounded.  The intersection of such
facets with the unit sphere gives the arcs of the intersection
diagram.  Our algorithm constructs the oriented boundary of the
intersection cell $\ID_P^\omega(p_i)$ by iterating over each facet of
the power cell $\Pow_P^\omega(p_i)$ and obtaining implicitly both the
set of vertices on the sphere $V$ and the set of oriented arcs
$E$. The oriented graph $G=(V,E)$ is called the \emph{boundary
  graph}. Once $G$ is constructed, the oriented cycles in the boundary
$\partial \ID_p^\omega(p_i)$ coincide with the (oriented) connected
components of $G$, and can be obtained by a simple traversal. In the
next paragraph, we explain the construction of $G=(V,E)$.

\subsection{Computation of the boundary graph}
We start by discarding the facets that do not contribute to the
intersection diagram.  We iterate over all facets $f$ of
$\Pow_P^\omega(p_i)$, and keep $f$ only if the predicate
\texttt{plane\_crossing\_sphere}$(f)$ returns $\infty$. (The case where
it returns $1$ corresponds to a trivial intersection between the facet
and the sphere and can be safely ignored.) Assuming that $f$ is not
discarded, the intersection of its supporting plane with the unit
sphere is a circle in $3$D.  If all vertices of $f$ are on or inside
the sphere (i.e., \texttt{has\_on}$\leq 0$), then $f$ is inside or only
touching the circle, and is discarded. See
Figure~\subref{fig:obtaining_arcs:inside}.

Given a facet $f = \Pow_P^\omega(p_i) \cap \Pow_P^\omega(p_j)$ that
was not discarded previously, we build a sequence of vertices on the
sphere by traversing $f$ in a clockwise sequence of ridges, oriented
with respect to the vector $p_j - p_i$~\footnote{We do not need orientation predicates at this point, since, in 
\textsc{CGAL}, $3$-simplices of a triangulation are positively oriented.}. 
We distinguish two different
kinds of vertices on the sphere, the \emph{\textbf{s}ource vertices} and the
\emph{\textbf{t}arget vertices}, as shown on
Figure~\subref{fig:obtaining_arcs:intersecting}. Given an oriented
ridge $r=[x_1,x_2]$, we compute the value of
\texttt{number\_of\_intersections(r)}. If the ridge has zero
intersections with the sphere, we skip it.  If there is only one
intersection between $r$ and the unit sphere, the corresponding vertex
is a source if $x_1$ is inside the sphere (i.e., \texttt{has\_on}(x\_1)
$\leq 0$) and it is a target vertex in the other case.  Finally, if
the ridge intersects the sphere twice, the oriented pair $(x_1,x_2)$
describes the target vertex, and the pair $(x_2,x_1)$ is the source
vertex. The arcs of the boundary graph $G$ corresponding to this facet
are obtained by matching each source vertex with a target vertex using
the same cyclic sequence. They are represented by an ordered pair of
vertex indices.

Finally, we need to consider the case where none of the ridges of $f$
intersect the sphere. There are two possibilities: either the facet
$f$ encloses the circle $\pi \cap \Sph^2$, where $\pi$ is the affine plane
spanned by $f$, or the circle is completely outside from $f$. In the former
case, we obtain an arc without vertex in the boundary of
$\ID_P^\omega(p_i)$, as shown in
Figure~\subref{fig:obtaining_arcs:enclosing}. Its orientation is
clockwise with respect to the vector $p_j - p_i$. To detect whether
this happens, we simply need to determine whether the center of the
circle lies inside or outside the convex polygon $f$. This is a
classical routine that can be performed using signed volumes. For this
test, we only need to use rational numbers because the center of the
circle is the projection of the origin on a plane described with
rational coordinates.

\begin{proof}[Proof of Theorem~\ref{thr:alg-complexity}]
First we build the regular triangulation of the set of $N$ weighted
points in $O(N \log N + C)$ time.  Running the algorithm described in
this section for every $p_i$ makes us visiting each $1$-simplex,
$2$-simplex, and $3$-simplex of the primal no more than two, three,
and four times respectively; and for each visit, it takes $O(1)$ time,
since we compute constant-degree predicates. Then we visit each
generated cycle at most twice. Therefore, the complexity of our
algorithm is in $O(N \log N + C + D)$, where $D$ is the complexity of
the diagram on the sphere. The fact that $D = O(C)$ completes the
proof.
\end{proof}




\section{Application to Minkowski-type problems involving paraboloids}
\label{sec:application}

This section deals with the far-field reflector problem
\cite{caffarelli2008weak}, which is a Minkowski-type problem involving
intersection of confocal paraboloids of revolution.  Consider a point
source light located at the origin $O$ of $\Rsp^d$ that emits lights
in all directions.  The intensity of the light emitted in the
direction $x \in \Sph^{d-1}$ is denoted by $\rho(x)$. Now, consider a
hypersurface $R$ of $\Rsp^d$. By Snell's law, every ray $x$ emitted by
the source light that intersects $R$ at a point whose normal vector is
$n$ is reflected in the direction $y=x-2\sca{x}{n}n$.  This means that
after reflection on $R$, the distribution of light at the origin given
by $\rho$ is transformed into a distribution $\nu_R$ on the set of
directions at infinity, a set that can also be described by the unit
sphere. The {\textit{far-field reflector problem}} consists in the
following inverse problem: given a density $\rho$ on the source sphere
$\Sph^2$ and a distribution $\mu$ on the sphere at infinity, the
problem is to find an hypersurface $R$ such that $\mu_R$ coincides
with $\mu$. For computational purposes, we assume that the target
measure $\nu$ is supported on a finite set of directions
$Y:=(y_i)_{1\leq i \leq N}$, and it is thus natural to consider a
reflector made of pieces of paraboloids with directions $Y$.  This
problem is numerically solved using an optimal transport formulation
due to \cite{wang2004design,glimm2003optical}.

\begin{figure}[t]
  \label{fig:monge_reflector}
  \centering          
  \subfloat[Initial $(\lambda_i)_{1\leq i \leq N}$]{\label{fig:initial-lambda_i}\includegraphics[width=0.3\textwidth]{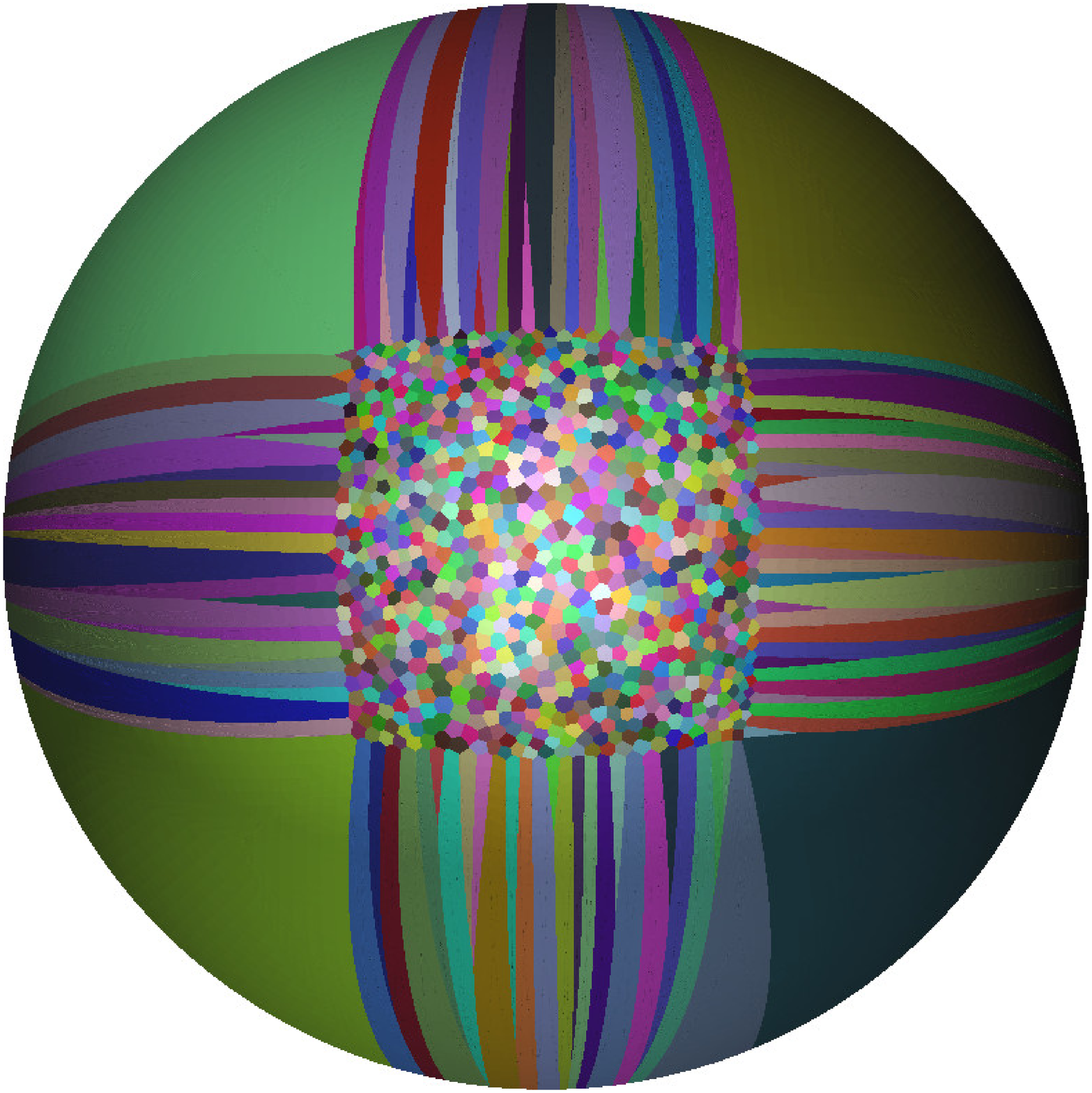}}
  \hspace{0.5cm}
 \subfloat[Final  $(\lambda_i)_{1\leq i \leq N}$]{\label{fig:final-lambda_i}\includegraphics[width=0.3\textwidth]{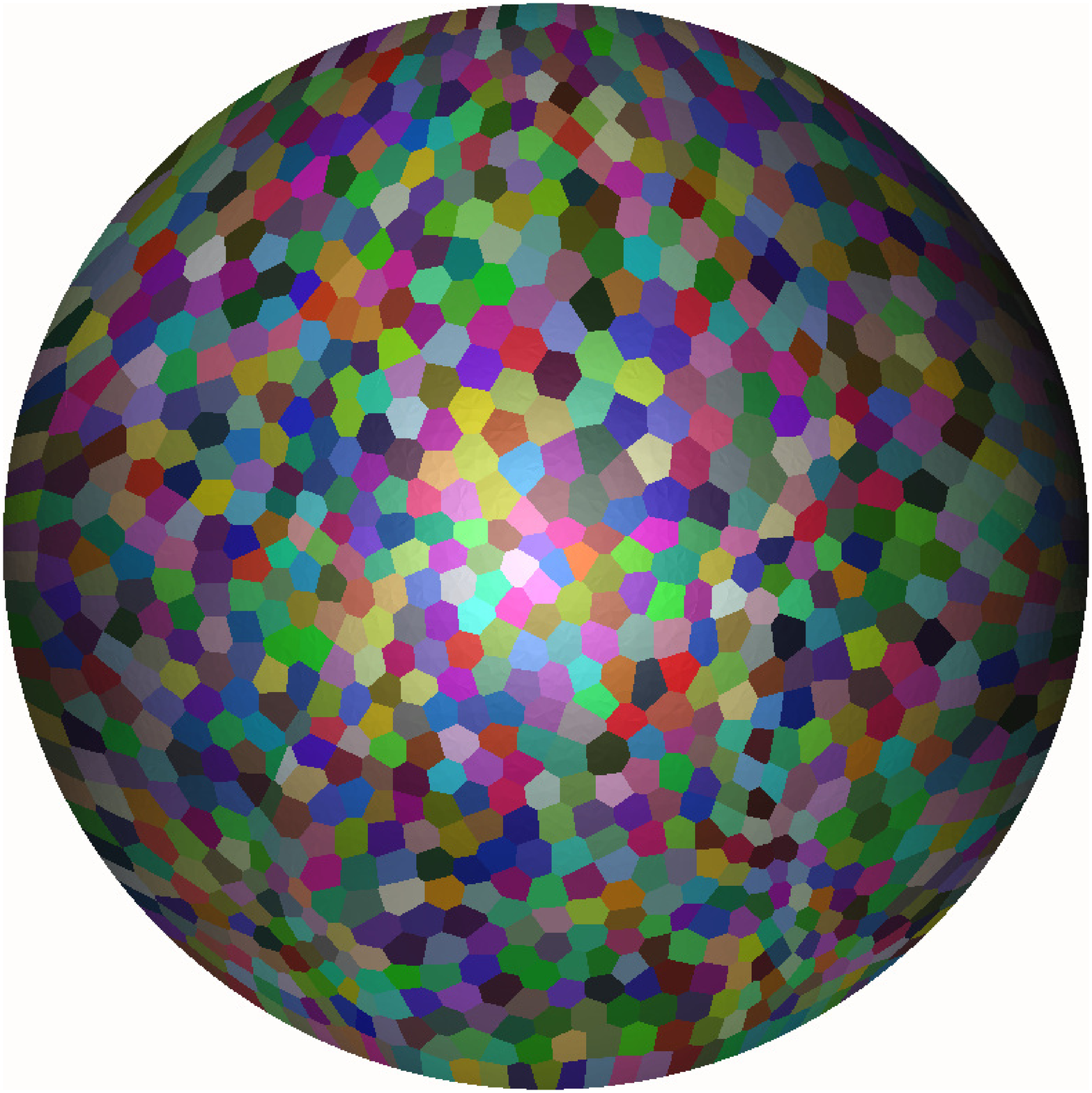}}
  \hspace{0.5cm}
 \subfloat[Rendering]{\label{fig:monge:1k}\includegraphics[width=0.3\textwidth]{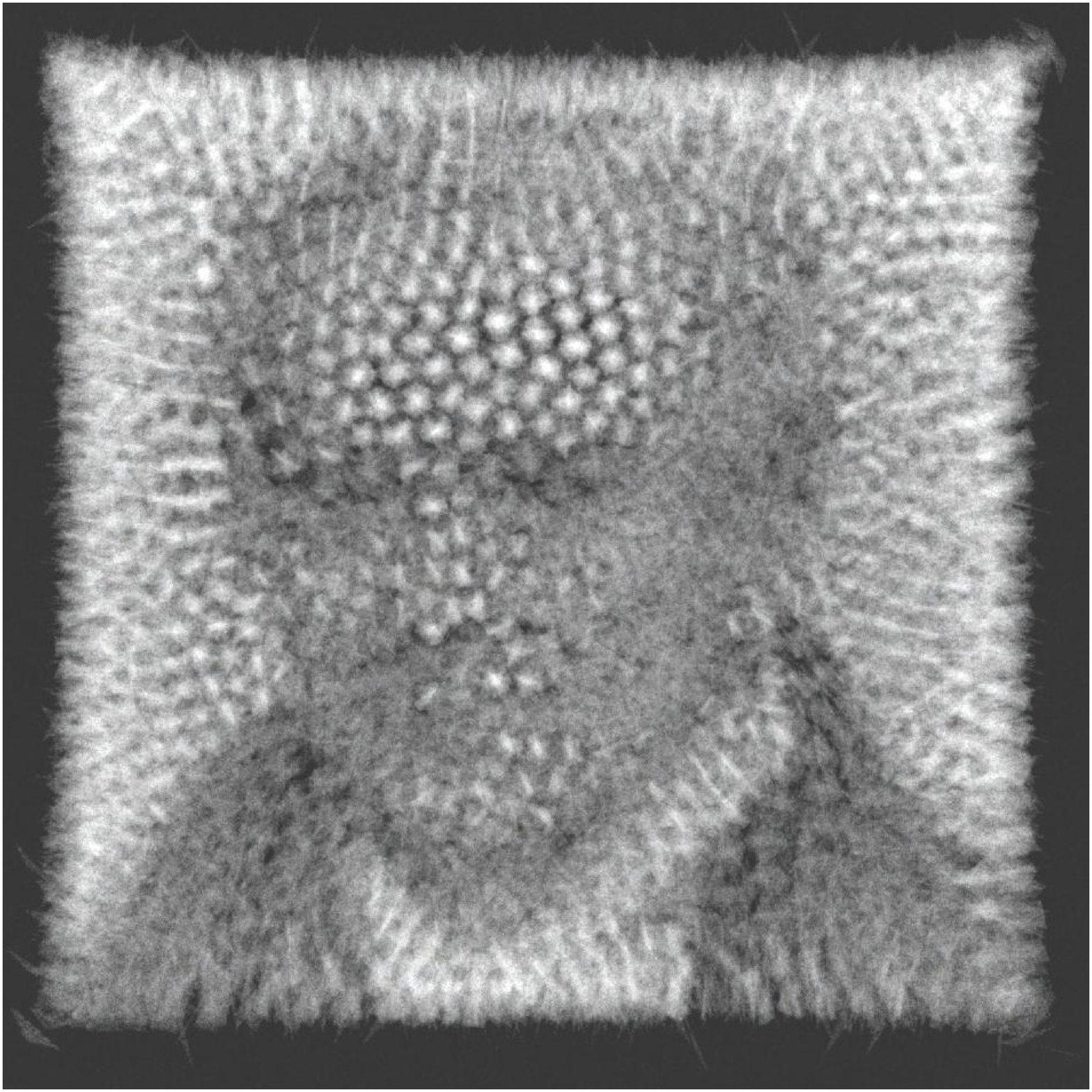}}  \\
  \subfloat[Final  $(\lambda_i)_{1\leq i \leq N}$]{\label{fig:monge:15k}\includegraphics[width=0.3\textwidth]{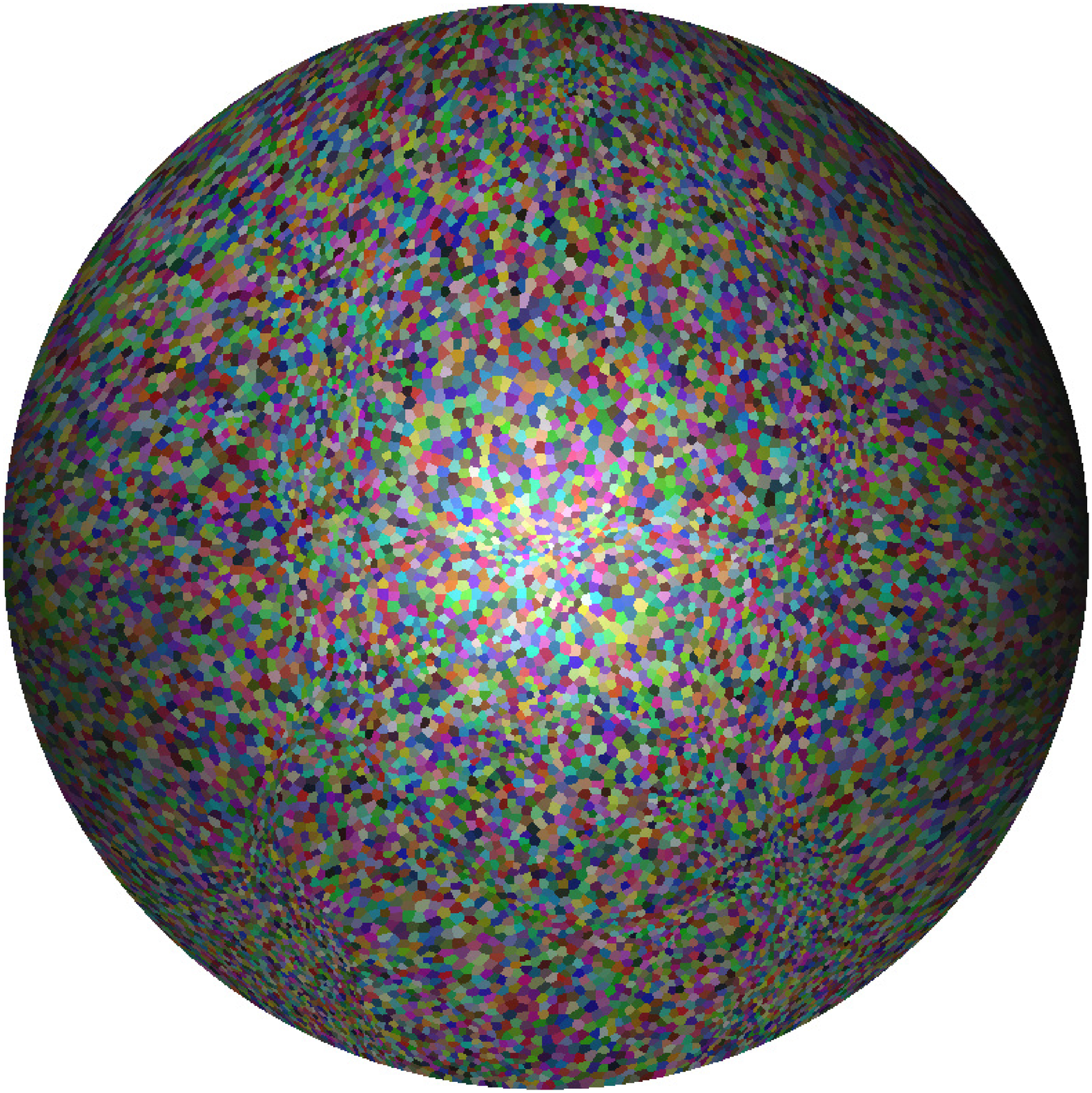}}
  \hspace{0.5cm}
 \subfloat[Reflector]{\label{fig:reflector}\includegraphics[width=0.3\textwidth]{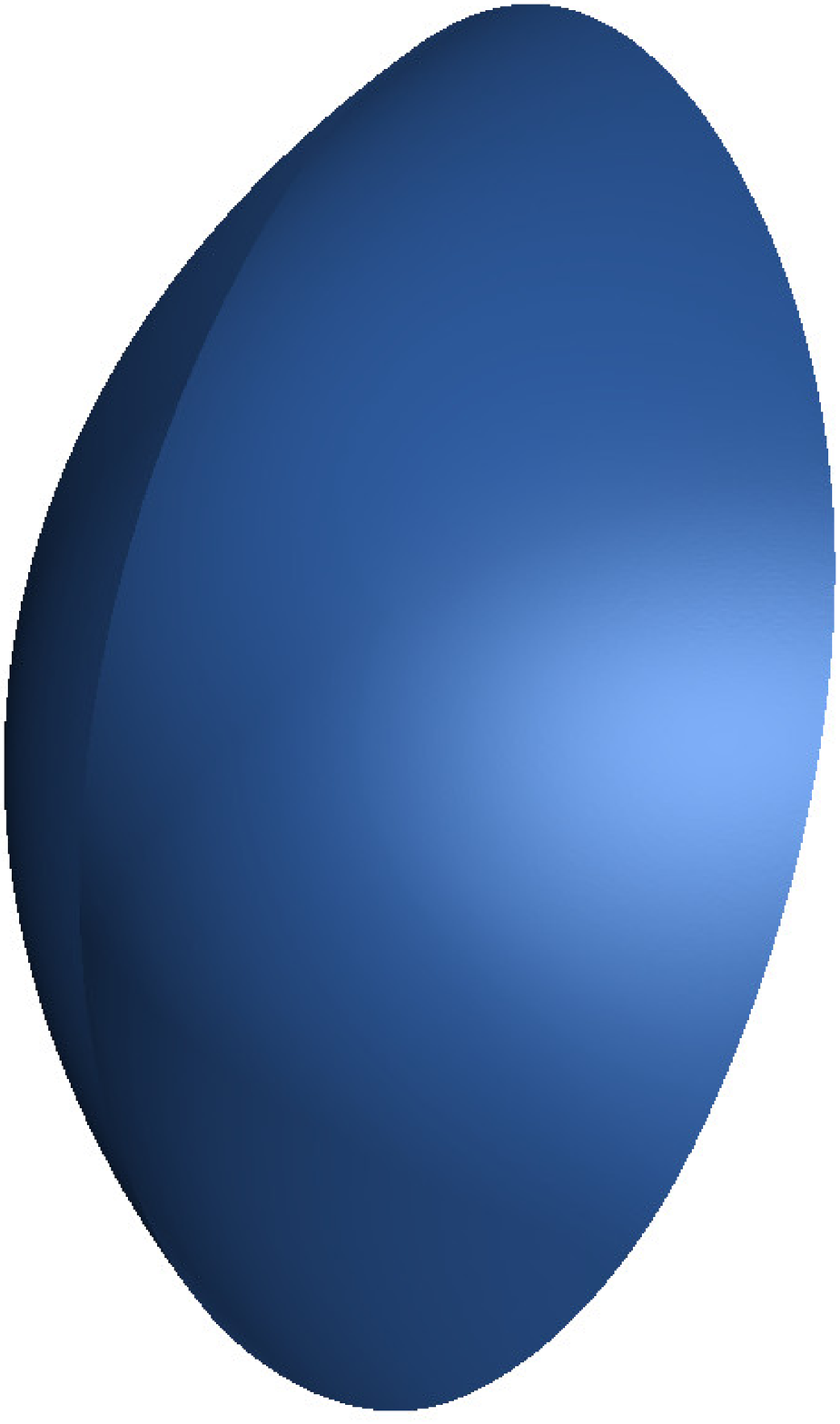}}
  \hspace{0.5cm}
 \subfloat[Rendering]{\label{fig:rendering}\includegraphics[width=0.3\textwidth]{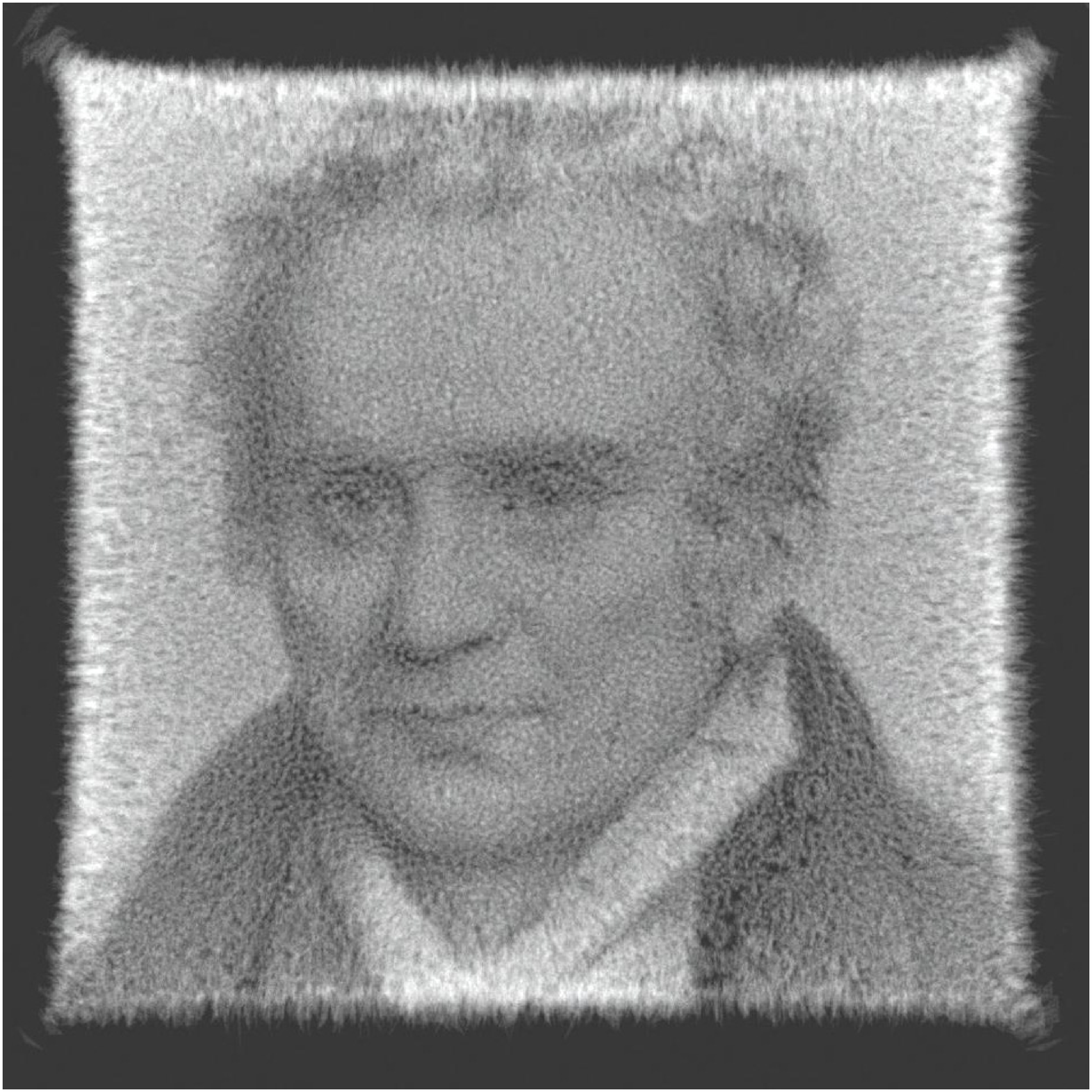}}  
\caption{\label{fig:monge}\textbf{Numerical computation.} {\small
    \emph{ Calculations were done with $N=1000$ paraboloids for the
      first row and $N=15000$ paraboloids for the second row. (a)
      Paraboloid intersection diagram for an initial
      $(\lambda_i)_{1\leq i \leq N}$. (b,d) Final intersection diagram
       after optimization. (e) Reflector surface defined by the
      intersection of paraboloids. (c,f) Simulation of the
      illumination at infinity from a punctual light source
      lighting uniformly $\Sph_-^2$, using \textsc{LuxRender}, a
 physically accurate raytracer engine.  
}}}
\end{figure}


\subsection{Far-field reflector problem.} We consider a finite family
$Y =(y_1,\hdots,y_N)$ of unit vectors that describe directions at
infinity, non-negative numbers $(\alpha_i)_{1\leq i\leq N}$ such that
$\sum_{1\leq i\leq N} \alpha_i = 1$ and a probability density $\rho$
on the unit sphere. The far-field reflector problem consists in
finding a vector of non-negative focal distances $(\lambda_i)_{1\leq
  i\leq N}$, such that
\begin{equation}
\forall i \in \{1,\hdots,N\},~~ \rho(\PI_Y^\lambda(y_i)) = \alpha_i,
\label{eq:FF}
\tag{FF}
\end{equation}
where for a subset $X$ of the sphere, $\rho(X):=\int_{X} \rho(u) \dd
u$ is the weighted area of $X$.  
The far-field reflector problem can be transformed into the
maximization of a concave functional, combining ideas from
\cite{wang2004design,glimm2003optical,aurenhammer1998minkowski}. 

\begin{theorem} 
\label{th:OT}
A vector of focal distances $(\lambda_i)_{1\leq i\leq N}$ solves the
far-field reflector problem \eqref{eq:FF} if and only if the vector
$(\gamma_i)_{1\leq i\leq N}$ defined by $\gamma_i = \log(\lambda_i)$
is a global maximizer of the following $\Class^1$ concave function:
\begin{equation}\Phi(\gamma) := \left[\sum_{i=1}^N \int_{\PI_Y^{\exp \gamma}(y_i)} (c(u,y_i) + \gamma_i) \rho(u)\dd u \right] - \sum_{i=1}^N \gamma_i \alpha_i 
\label{eq:Phi}
\end{equation}
where $c(u,v) := - \log(1-\sca{u}{v})$ and with the convention
$\log(0)= -\infty$. Moreover, the gradient of the function $\Phi$ is given by
\begin{equation}
\nabla\Phi(\gamma) :=(\rho(\PI_Y^{\exp(\gamma)}(y_i)) - \alpha_i)_{1\leq i\leq
  N}.\label{eq:gradPhi}
\end{equation}
\end{theorem}

There is a similar formulation for a reflector defined by the union of
solid confocal paraboloids\,\cite{glimm2003optical}. However, there is
no known variational formulation for the reflector problem involving
intersection of ellipsoids. A numerical approach has been proposed in
\cite{oliker2006rigorous}.

We now turn to the proof of Theorem~\ref{th:OT}. This proof combines the
results of Section~5 of \cite{aurenhammer1998minkowski} with the
optimal transport formulation of the far-field reflector problem
\cite{wang2004design,glimm2003optical}.
Let us first recall some properties of supdifferential of concave functions.
Given a function $\Phi$ and $\lambda$ in $\Rsp^N$, the supdifferential
of $\Phi$ at $\lambda$, denoted $\partial^+ \Phi(\lambda)$ is the set of
vectors $v$ such that
\[ \forall \kappa \in \Rsp^N,~ \Phi(\kappa) \leq \Phi(\lambda) +
\sca{\kappa-\lambda}{v}.\] A function $\Phi$ is concave if and only if
for every $\lambda$, the supdifferential $\partial^+ \Phi(\lambda)$ is
nonempty. If this is the case, $\Phi$ is differentiable almost
everywhere, and at points of differentiability the supdifferential
$\partial^+ \Phi(\lambda)$ coincides with the singleton $\{\nabla
\Phi(\lambda)\}$. Finally, $\lambda$ is a global maximum of $\Phi$ if
and only if $\partial^+ \Phi (\lambda)$ contains the zero vector.

\begin{proof}[Proof of Theorem~\ref{th:OT}] We consider a vector $\gamma$ in $\Rsp^N$. First remark that $c(u,y_i) + \gamma_i =
\log(\exp(\gamma_i)/(1-\sca{u}{y_i})$. Therefore,
\begin{align*}
u \in \PI_Y^{\exp(\gamma)}(y_i)
&\Longleftrightarrow \forall j\in\{1,\hdots,N\},
\frac{\exp(\gamma_i)}{1 - \sca{u}{y_i}} \leq 
\frac{\exp(\gamma_j)}{1 - \sca{u}{y_j}} \\
&\Longleftrightarrow  \forall j\in\{1,\hdots,N\},
c(u,y_i) + \gamma_i \leq c(u,y_j) + \gamma_j
\end{align*}
Therefore, the function $\Phi$ can be reformulated as follows:
\[\Phi(\gamma) = \int_{\Sph^{d-1}} \left(\min_{1\leq i\leq N}
  c(u,y_i) + \gamma_i\right)\rho(u) \dd u - \sum_{i=1}^N \gamma_i
  \alpha_i.\]
We now define $T_\gamma$ as the function
that maps a point $u$ on the unit sphere to the point $y_i$ such that $u$ belongs to $\PI_Y^{\exp(\gamma)}(y_i)$. Then,
\begin{equation}
\Phi(\gamma) = \int_{\Sph^{d-1}} (c(u,T_\gamma(u)) +
  \gamma_{T_\gamma(u)})\rho(u) \dd u - \sum_{i=1}^N \gamma_i
  \alpha_i.
\label{eq:i}
\end{equation}
Moreover, for any $\kappa$ in $\Rsp^d$, one has
$\min_{1\leq i\leq N} c(u,y_i) + \kappa_i \leq c(u,T_\gamma(u)) +
  \kappa_{T_\gamma(u)}$.  Integrating this inequality, and substracting
  \eqref{eq:i}, we get

\begin{align}
\Phi(\kappa) - \Phi(\gamma)
&\leq \int_{\Sph^{d-1}} (\kappa_{T_\gamma(u)} - \gamma_{T_\gamma(u)})\rho(u)\dd u - \sum_{1\leq i\leq N} (\kappa_i - \gamma_i) \alpha_i \notag\\
&\leq \sum_{1\leq i\leq N} \left(\int_{\PI_Y^{\exp\gamma}(y_i)} \rho(u)\dd u - \alpha_i\right) (\kappa_i - \gamma_i) = \sca{D\Phi(\gamma)}{\kappa - \gamma} \notag\\
\hbox{where }
D\Phi(\gamma) &:=(\rho(\PI_Y^{\exp(\gamma)}(y_i)) - \alpha_i)_{1\leq i\leq
  N}.\label{eq:DPhi}
\end{align}
The above inequality shows that $D\Phi(\gamma)$ lies in
$\partial^+\Phi(\gamma)$, i.e., this set is never empty and the
function $\Phi$ is concave. Since the vector $D\Phi(\gamma)$ depends
continuously on $\gamma$, we deduce that the function $\Phi$ is
$\Class^1$ smooth, and that $\nabla \Phi(\gamma) = D\Phi(\gamma)$
everywhere. By Equation~\eqref{eq:DPhi}, a vector $\lambda :=
\exp(\gamma)$ solves the far-field reflector problem \eqref{eq:FF} if
and only if $D\Phi(\gamma)=0$, i.e., if and only if $\gamma$ is a
global maximizer of $\Phi$.
\end{proof}

\subsection{Implementation details.} The implementation of the maximization 
of the functional $\Phi$ follows closely \cite{merigot2011multiscale}.
We rely on a quasi-Newton method, which only requires  being able to
evaluate the value of $\Phi$ and the value of its gradient at any
point $\gamma$, as given by Equations~\eqref{eq:Phi}--\eqref{eq:gradPhi}.
The computations of these values are performed in two steps. First, we
compute the boundary of the paraboloid intersection cells
$\PI_Y^{\exp(\gamma)}(y_i)$, using the algorithm described in
Section~\ref{sec:computation}. These cells are then tessellated, and
the integrals in Equations~\eqref{eq:Phi}--\eqref{eq:gradPhi} are
evaluated numerically using a simple Gaussian quadrature.  In the
experiments illustrated in Figure \ref{fig:monge}, we constructed the
measure $\sum_{i} \alpha_i \delta_{y_i}$ so as to approximate a
picture of Gaspard Monge (projected on a part of the half-sphere
$\Sph^2_+ := \Sph^2\cap \{z\geq 0\}$).  The density $\rho$ is
constant in the half-sphere $\Sph^2_-$ and vanishes in the other half.
To the best of our knowledge, the only other numerical implementation
of this formulation of the far-field reflector problem has been
proposed in \cite{caffarelli1999numerical}. The authors
develop an algorithm, called {\textit{Supporting paraboloids}} 
which bears resemblance to Bertsekas' auction algorithm for the
assignment problem \cite{bertsekas1988dual}. They use it to solve the far-field reflector
problem with 19 paraboloids. Using the quasi-Newton approach presented
above, and the algorithm developed in Section~\ref{sec:computation},
we are able to solve this problem for 15,000 paraboloids in less than
10 minutes on a desktop computer. Note that the algorithm of
Section~\ref{sec:computation} would probably also allow a faster and
robust implementation of the \textit{Supporting paraboloids} algorithm \cite{caffarelli1999numerical}.

\section{Conclusion} \label{sec:conclusion} In addition to the open
problems mentioned earlier, let us mention a perspective.  In a recent
article \cite{kitagawa2012iterative}, the algorithm of supporting
paraboloids was extended to optimal transport problems involving a
cost function $c$ that satisfies the so-called Ma-Trudinger-Wang
regularity condition.  For this algorithm to be practical, one needs
to compute the generalized Voronoi cells efficiently, defined for
any function $\psi:Y\to\Rsp$ by
\[\Vor_c^\psi(y) = \{ x \in X;~\forall z \in Y,~c(x,y) + \psi(y) \leq c(x,z)+\psi(z)\}.\]
For general costs, and even in $2$D, one cannot hope to do this in
time below $\Omega(N^2)$. However, the MTW regularity condition ensures an
analog of Proposition~\ref{prop:lemma-2-1} and in particular, it
implies that these generalized Voronoi cells are connected.  One might
wonder whether a randomized iterative construction could be used in
this setting to yield a construction in expected time $O(N\log N)$ in
$2$D. This would open the way to practical algorithms for the
resolution of optimal transport problems that are intractable to
PDE-based methods.

\subsection{Acknowledgements.} The authors would like to thank
Dominique Attali, Olivier Devillers and Francis Lazarus for
interesting discussions. Olivier Devillers suggested the approach used
in the proof of the lower complexity bound for intersection of
ellipsoids. The first author is supported by grant
FACEPE/INRIA,APQ-0055-1.03/12. The second and third author would like
to acknowledge the support of the French Agence Nationale de la
Recherche (ANR) under reference ANR-11-BS01-014-01 (TOMMI) and
ANR-13-BS01-0008-03 (TOPDATA) respectively.

\bibliographystyle{amsplain}
\bibliography{ot-reflector}

\end{document}